\include{birkart.sty}

\documentclass[10pt]{amsart}
\pdfoutput=1

\input txdtools

\usepackage{amsmath,mathtools}
\usepackage{amssymb}
\usepackage{mathrsfs}
\usepackage{amsfonts}
\usepackage{graphicx}
\usepackage{texdraw}
\usepackage{graphpap}
\usepackage{enumitem}
\usepackage{chngcntr}
\usepackage{hyperref}
\mathtoolsset{showonlyrefs}

\usepackage{wasysym}
\usepackage{color}
\usepackage[OT2,T1]{fontenc}
\usepackage{verbatim}

\def\bfK{\mathbf{K}}

\def\bfR{\mathbf{R}}

\def\C{\mathbb{C}}

\def\R{\mathbb{R}}

\def\RH{\mathbb{H}_{+}}

\def\calK{\mathcal{K}}

\def\calD{\mathcal{D}}
\def\calH{\mathscr{H}}

\def\calH{\mathcal{H}}

\def\calD{\mathcal{D}}

\def\erfc{\operatorname{erfc}}

\def\I{\mathrm{I}}
\def\II{\mathrm{II}}
\def\III{\mathrm{III}}
\def\Lap{\Delta}
\def\normal{\mathrm{n}}

\newcommand{\Prob}{{\mathbf P}}

\newcommand{\1}{{\mathbf{1}}}
\newcommand{\re}{\operatorname{Re}}
\newcommand{\im}{\operatorname{Im}}

\newcommand{\E}{{\mathbf E}}

\renewcommand{\d}{{\partial}}
\newcommand{\dbar}{\bar{\partial}}

\theoremstyle{plain}
\newtheorem{thm}{Theorem}[section]
\newtheorem{lem}[thm]{Lemma}

\theoremstyle{definition}
\newtheorem*{def*}{Definition}
\newtheorem*{eg*}{Example}

\theoremstyle{remark}
\newtheorem*{rmk*}{Remark}

\numberwithin{equation}{section}

\begin{document}

\title[]{Edge behavior of two-dimensional Coulomb gases near a hard wall}

\author{Seong-Mi Seo}
\address{Seong-Mi Seo\\
School of Mathematics\\
Korea Institute for Advanced Study\\
85 Hoegiro\\
Dongdaemun-gu\\
Seoul 02455\\
Republic of Korea}
\email{seongmi@kias.re.kr}




\begin{abstract} 
We consider a two-dimensional Coulomb gas confined to a disk when the external potential is radially symmetric. 
In the presence of a hard-wall constraint effective to change the equilibrium, the density of the equilibrium measure acquires a singular component at the hard wall. In the determinantal case, we study the local statistics of Coulomb particles at the hard wall and prove that their local correlations are expressed in terms of "Laplace-type" integrals, which appear in the context of truncated unitary matrices in the regime of weak non-unitarity. 
\end{abstract}

\maketitle

\section{Introduction}


In this paper, we are interested in the local statistics of the two-dimensional Coulomb gas model which consists of $N$ particles interacting on the complex plane $\C$ via a repulsive Coulomb potential. The system of Coulomb particles with an external potential $Q$ is defined by the Gibbs distribution 
\begin{equation}\label{Gbs}
\Prob_{N,\beta}(\zeta_1,\cdots, \zeta_N) = \frac{1}{Z_{N,\beta}} e^{-\beta H_N(\zeta_1,\cdots,\zeta_N)},\quad \zeta_1,\cdots,\zeta_N \in \C,
\end{equation}
where $\beta>0$ denotes the inverse temperature and
$H_N(\zeta_1,\cdots,\zeta_N)$ is the Hamiltonian of the system at position $(\zeta_1,\cdots,\zeta_N)\in \C^N$,
given by
\begin{equation*}
H_{N}(\zeta_1,\cdots, \zeta_N) = -\sum_{j\ne k}\log|\zeta_k-\zeta_j| + N\sum_{j=1}^{N}Q(\zeta_j).
\end{equation*}
Here, $Z_{N,\beta}$ is the normalizing constant called the partition function of the system. For the special choice $\beta=1$, the system has a determinantal structure and has a close relationship with some non-Hermitian random matrix ensembles. When $\beta=1$, the Gibbs distribution in \eqref{Gbs} coincides with the joint distribution of eigenvalues in random normal matrix models (see e.g., \cite{CZ} and \cite{EF05}) and in the case of $Q(\zeta)=|\zeta|^2$, the particle system corresponds to the eigenvalues of complex Ginibre matrices, which are $N \times N$ non-Hermitian matrices whose entries are i.i.d. standard complex Gaussians \cite{Gin}.

If the external potential grows sufficiently fast near infinity, the particles tend to accumulate on a compact set, which is called the droplet. More precisely, as $N\to \infty$, the empirical measure $\mu_N = N^{-1}\sum_{j=1}^{N}{\delta_{\zeta_j}}$ for the system $\{\zeta_j\}$ converges to an equilibrium measure with compact support, which minimizes the weighted logarithmic energy
\begin{align}\label{Len}
I_Q[\mu]:= \int\!\!\!\int_{\C^2}\log\frac{1}{|\zeta-\eta|}d\mu(\zeta)d\mu(\eta)+\int_{\C}Q(\zeta)d\mu(\zeta)
\end{align}
among all probability measures on $\C$. See \cite{HM, HP98}.
Central limit theorems for the fluctuations of the linear statistics have been proved in \cite{RV} for Ginibre ensemble and \cite{AHM2, AHM3} for random normal matrix ensembles,  i.e., two-dimensional Coulomb gases at $\beta =1$. These results have been extended to general $\beta$ \cite{BNY, LS}. 

Microscopic behaviors of Coulomb gases are relatively less known.  Exact results are known for $\beta = 1$, the determinantal case, due to the correlation structure expressed by orthogonal polynomials.
In the bulk of the droplet, the support of the equilibrium measure, universality of Ginibre kernel 
\begin{equation}\label{Ginker}
G(z,w) = e^{z\bar{w}-|z|^2/2-|w|^2/2}\end{equation}
 has been obtained from an asymptotic expansion for the reproducing kernel of  polynomial Bergman spaces. See \cite{AHM2, Ber}. At a regular boundary point of the droplet,  universality for the kernel 
\begin{equation}\label{K:free}
G(z,w) F(z+\bar{w}),
\end{equation}
where $G$ is the Ginibre kernel \eqref{Ginker} and $F$ is the free boundary plasma function defined by
$$F(z) = \frac{1}{2} \erfc\Big(\frac{z}{\sqrt{2}}\Big) 
=\frac{1}{\sqrt{2\pi}} \int_{-\infty}^{0} e^{-(z-t)^2/2} dt,$$
has been proved in \cite{HW17} by obtaining an asymptotic expansion for orthogonal polynomials near the boundary. For general $\beta$, local densities for two-dimensional Coulomb gases have been studied in for instance \cite{BNY17, L17}. 

If a boundary confinement is imposed to the Coulomb gas system, edge behaviors of the particles may change. One may consider Coulomb gases forced to be in a certain region of the plane. This kind of confinement can be achieved by redefining the external potential to be $+\infty$ outside the region. A Coulomb gas system with a hard edge, the case when the region of confinement is exactly the droplet, was studied in several literatures \cite{Fo, HM, Sm}. In this case, the equilibrium measure does not change, but the local statistics at the hard edge are described by a different kernel. At $\beta = 1$, the Coulomb gas system properly centered and rescaled at the hard edge converges to a determinantal point process with the kernel
\begin{equation}\label{K:hard}
G(z,w) H(z+\bar{w}),\quad \re z ,\,\re w \leq 0, 
\end{equation} where $G$ is the Ginibre kernel and $H$ is the hard edge plasma function 
$$H(z)= \frac{1}{\sqrt{2\pi}}\int_{-\infty}^{0} \frac{e^{-(z-t)^2/2}}{F(t)} dt,\quad \re z \leq 0$$
for a class of external potentials containing radially symmetric ones \cite{A18, AKM}. A family of boundary confinements which interpolates between a free boundary \eqref{K:free} and a hard edge \eqref{K:hard} has been studied in \cite{AKS}.

\begin{figure}
\includegraphics[width=.32\textwidth]{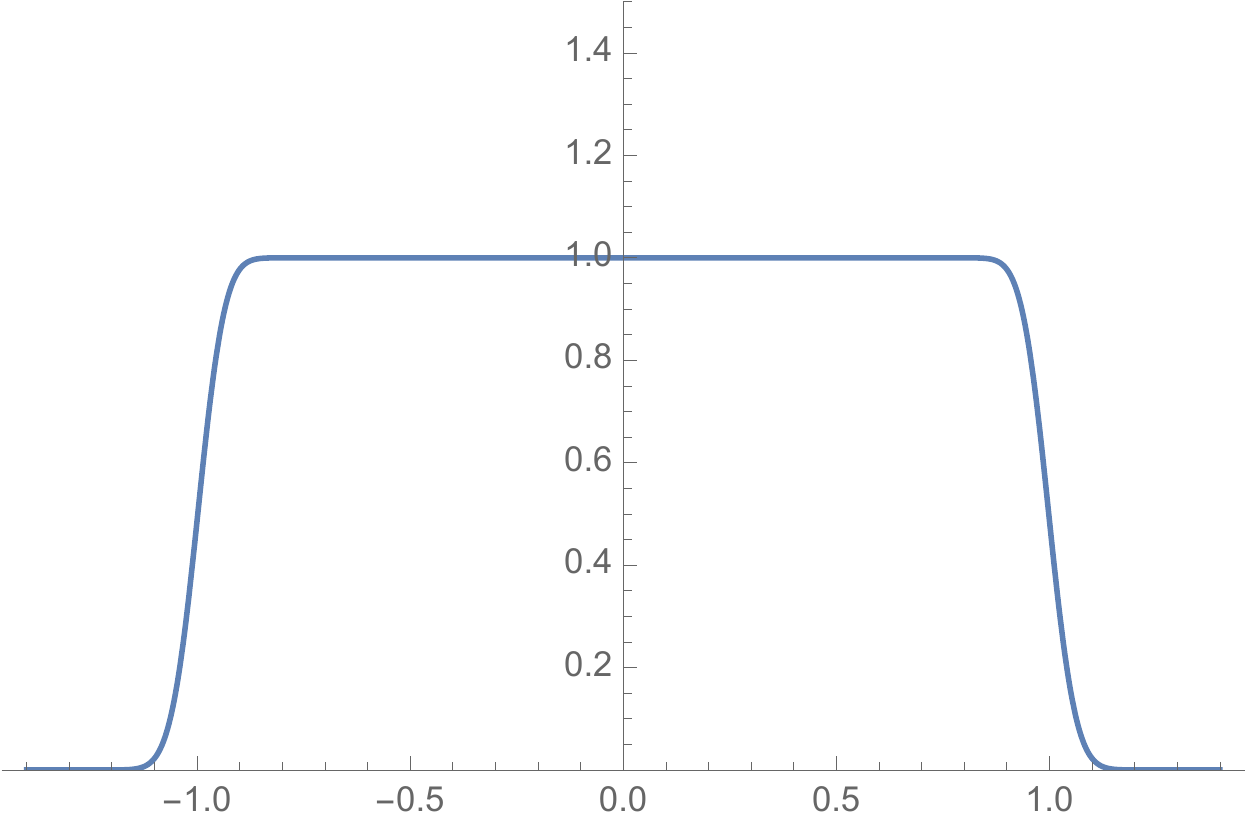}
\includegraphics[width=.32\textwidth]{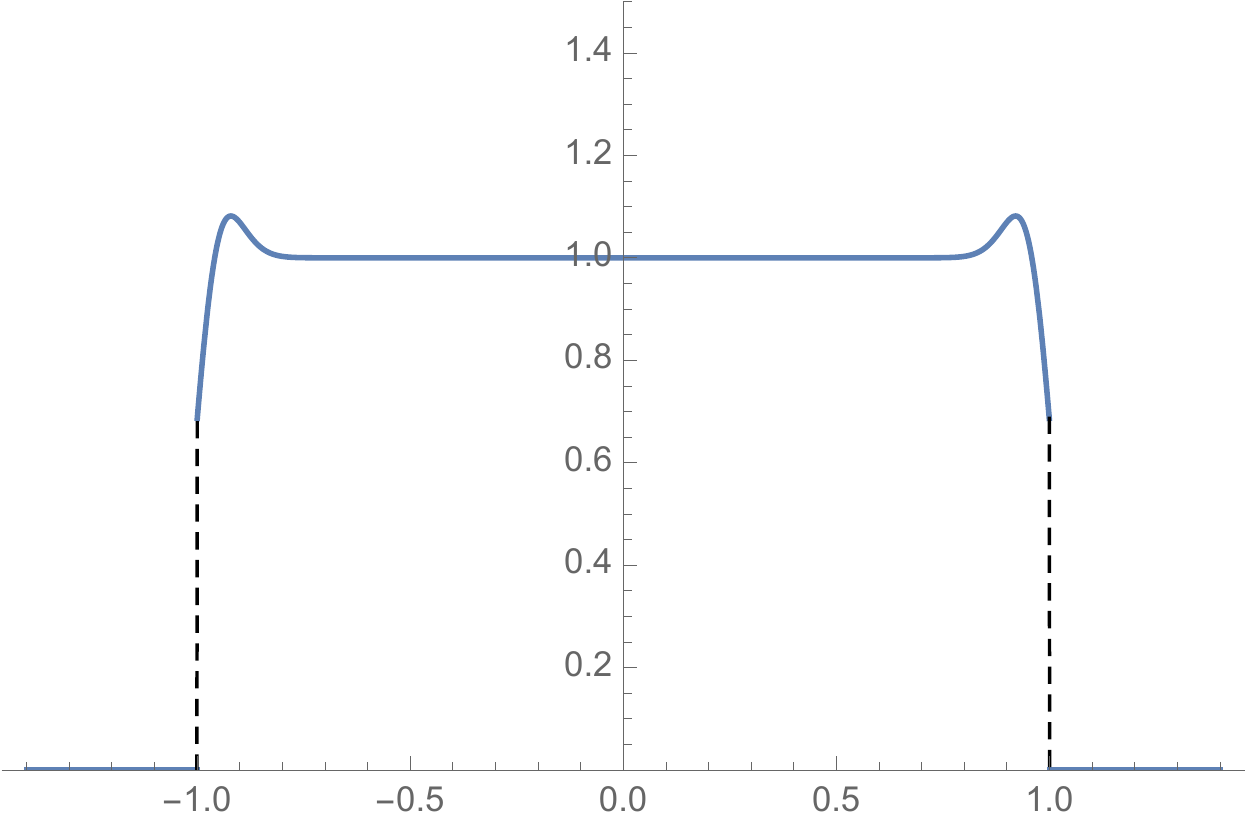}
\includegraphics[width=.32\textwidth]{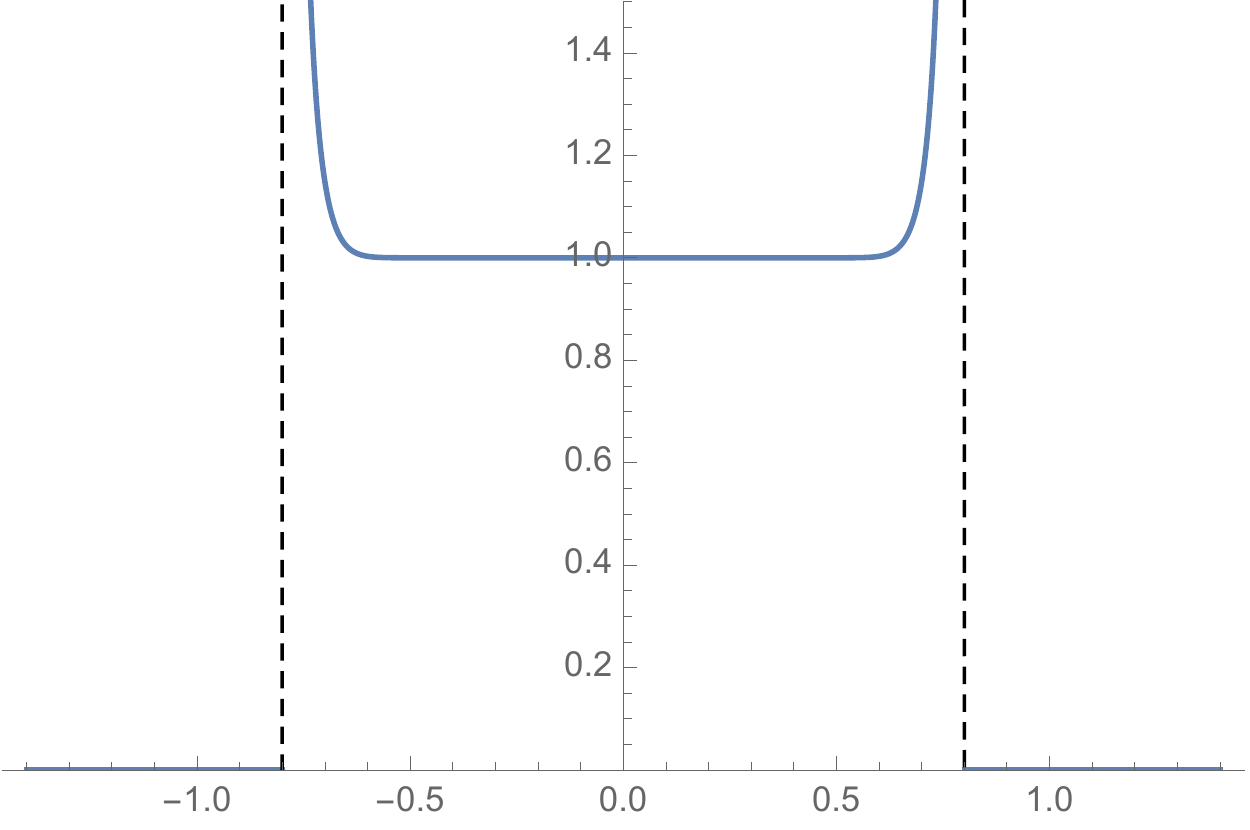}
\caption{Density profiles of Ginibre ensemble with a free boundary (left) and confined in a unit disk (middle) and in a disk of radius $r=0.8$ (right)  when $n=100$.} \label{fig:R}

\end{figure}

In this paper, we will study the Coulomb gases confined in a proper subset of the droplet. Figure \ref{fig:R} shows the density profiles of the Ginibre ensemble $(Q(\zeta)=|\zeta|^2)$ with different boundary confinements. It is well known that the empirical measure for Ginibre ensemble converges to 
the uniform measure on the unit disk. In Figure \ref{fig:R}, the left one shows the density profile of the particles without any constraints on the boundary, where the local behavior of the system near the edge $|\zeta|=1$ is expressed in terms of the free boundary function $F$ in \eqref{K:free}.
The middle one is for the hard edge case where the particles are confined in the droplet (the unit disk) and the local behavior is expressed in terms of the hard edge plasma function $H$ in \eqref{K:hard}. If the Coulomb particles are confined in a smaller disk with radius $r<1$, then a "hard wall" exists along the boundary of the disk and the equilibrium measure changes. In this case, the local behavior is no longer described by the "error function" type kernels like \eqref{K:free} or \eqref{K:hard}. We shall show that the Coulomb particle system at $\beta=1$, properly rescaled at the hard wall in the inward normal direction,
converges to the determinantal point process with kernel represented as a Laplace-type integral
\begin{equation}\label{K:hw}
K^{(0)}(z,w)= \int_{0}^{1} t \,e^{-t(z+\bar{w})} \,dt 
,\quad \re z, \re w > 0
\end{equation}
and this kernel is universal for a class of radially symmetric potentials. Furthermore, under a small perturbation in the external potential defined by adding a logarithmic singularity at the boundary $- \frac{\alpha}{N}\log(r-|\zeta|)$
for $\alpha>-1$, the limiting correlation kernel is of the form 
\begin{equation}\label{K:log}
K^{(\alpha)}(z,w) = (2\re z)^{\alpha/2}(2\re w)^{\alpha/2}\int_{0}^{1} \frac{t^{\alpha+1}}{\Gamma(\alpha+1)}e^{-t(z+\bar{w})} \,dt,\quad \re z,\re w > 0.
\end{equation}

The limiting kernel \eqref{K:log} agrees with the kernel found for truncated unitary matrices introduced in \cite{ZS}. More precisely, a non-Hermitian matrix $J$ which is the upper left block of a unitary matrix $U$ taken randomly from the unitary group $U(M)$ was studied in \cite{ZS}. The eigenvalues of the matrix $J$ are located in the unit disk, and in the case when the size of the truncation is $N\times N$ with $N=M-1-\alpha$ for a non-negative integer $\alpha$, the system of eigenvalues properly rescaled at the circle in the inward normal direction converges to the determinantal point process with correlation kernel \eqref{K:log} as $N\to \infty$. This agreement is due to the fact that the joint distribution of the eigenvalues for truncated unitary matrices agrees with that of two-dimensional Coulomb gases confined in a unit disk with a hard wall at the unit circle, especially when the external potential is $Q(z) = -\frac{\alpha}{N}\log(1-|z|^2)$ for $|z|\leq 1$. 
An application of this model in physics can be found in the theory of quantum chaotic scattering. See \cite{Fyo99, Fyo03} for the random matrix approach to the chaotic scattering and results on random contractions and deformations of Hermitian matrices.

The outline of the paper is the following. In Section \ref{sec:main} we introduce two-dimensional Coulomb gas ensembles with a hard wall and state the main results. In Section \ref{sec:hard} we compute edge scaling limits at the hard wall by approximating weighted orthonormal polynomials. Section \ref{sec:ward} deals with a rescaled version of Ward's equation, which gives an abstract approach to universality for the kernel \eqref{K:log}.

\subsection*{Notation}
$\RH = \{z\in \C : \re z > 0\}$ is the right half plane and $\R_{+} = \{x\in \R : x>0\}$ is the positive real line. 
$D(p,r)$ is the open disk with center $p$ and radius $r$. We write $\d$ and $\bar{\d}$ for the usual complex derivatives and $\Lap = \d \bar{\d}$ for the usual Laplacian on $\C$ divided by $4$. $\1_{U}$ denotes the characteristic function of a set $U$. $dA = dxdy/\pi$ denotes the Lebesgue measure on $\C$ divided by $\pi$ and $ds(\zeta)=\frac{1}{2\pi}|d\zeta|$ denotes the arclength measure divided by $2\pi$. We use the symbol $\delta_N$ for the number $\delta_N = N^{-\frac{1}{2}}\log N$.

\section{Coulomb gas ensembles with a hard wall}\label{sec:main}

\subsection{Equilibrium measures and droplets}

Consider a potential $Q:\C \to \R\cup \{+\infty\}$ which is lower semi-continuous and $Q < +\infty$ on a set of positive area. Suppose that $Q$
satisfies the growth condition 
\begin{equation}\label{Qgr}
\liminf_{|\zeta|\to +\infty} Q(\zeta)/\log|\zeta|^2 > 1.\end{equation} This assumption leads to the existence of the equilibrium measure $\sigma$, the unique probability measure which minimizes the weighted logarithmic potential energy 
$$I_Q[\mu]:= \int\!\!\!\int_{\C^2}\log\frac{1}{|\zeta-\eta|}d\mu(\zeta)d\mu(\eta)+\int_{\C}Q(\zeta)d\mu(\zeta)$$
among all positive unit Borel measures on $\C$. The measure $\sigma$ has a compact support $S$ called the droplet. 
More details can be found in \cite[Section 1]{ST}. 

For a parameter $\tau>0$ we consider the potential $\tau^{-1}Q$ and the corresponding weighted logarithmic energy problem. Under the assumptions on $Q$ above, for $\tau$ with $0<\tau\leq 1$ there exists the unique equilibrium measure $\sigma_\tau$ associated with the external potential $\tau^{-1}Q$. We define $S_\tau$ by the support of the measure $\sigma_\tau$, called the $\tau$-droplet. 
One can consider $\tau$-droplets for $\tau>1$. For this, the growth assumption $\liminf_{|\zeta|\to +\infty} Q(\zeta)/\log|\zeta|^2 > \tau$ instead of \eqref{Qgr} yields the existence and compactness of $S_\tau$. It is well known that $S_\tau$ increase with $\tau$ and evolves according to Laplacian growth. We refer to \cite{HM, HW17, Z} for more results and discussions on this topic.


Throughout the paper, we suppose that $Q$ is radially symmetric, write $Q(\zeta)=q(r)$ for $r=|\zeta|$. 
Under the assumption that $Q$ is subharmonic on $\C$ and $q$ is differentiable on $\R^{+}$ with absolutely continuous derivative, the droplet $S=S_1$ is the ring 
$$\{\zeta \in \C : \rho_0 \leq |\zeta | \leq \rho_1 \}$$
where $\rho_0$ is the smallest number such that $q'(r) > 0$ for all $r>\rho_0$ and $\rho_1$ is the smallest number such that $\rho_1q'(\rho_1)=2$. Moreover, the equilibrium measure $\sigma$ is of the form 
\begin{equation}\label{Oreq}
d\sigma(\zeta) = \Lap Q(\zeta) \,\1_{S}(\zeta) \,dA(\zeta) = \frac{1}{4\pi}(rq'(r))' \1_{S} dr d\theta, \quad \zeta = re^{i\theta}.
\end{equation}
Here, $dA$ denotes the two-dimensional Lebesgue measure divided by $\pi$ and $\1_{U}$ denotes the indicator function for a subset $U \subset \C$. 
See \cite[Section 4.6]{ST} for the proof.

For $\rho_0 \leq r \leq \rho_1$, the total mass of the measure $\sigma$ restricted to $S\cap D(0,r)$ is $\frac{1}{2}rq'(r)$, which is increasing and takes values in $[0,1]$. Also note that for given $\tau$ with $0<\tau \leq 1$, the droplet $S_\tau$ is the ring 
$$S_\tau = \{\zeta\in \C : \rho_0 \leq |\zeta| \leq \rho_\tau \},$$ where $\rho_\tau$ is the smallest number such that $\rho_\tau q'(\rho_\tau)=2\tau$.

\subsection{Localization}\label{sec:loc} 
For a given potential $Q$ and a disk $\calD =D(0,{\rho_{*}})$ of center $0$ and radius $\rho_{*}$, we define the localized potential 
$$Q^{\calD} = Q\cdot \1_{\calD} + \infty \cdot \1_{\C \setminus \calD}.$$ The particles picked randomly with respect to the Gibbs measure \eqref{Gbs} under the external potential $Q^{\calD}$ are completely confined to the set $\calD$, i.e., this localization of the external potential defines the Coulomb gas system constrained to $\calD$.  

In general, this type of localization can be considered for any set of positive area. The notion of local droplet was introduced in \cite[Section 5]{HM} for the droplet with the localization of potentials. In \cite{A18,AKMW}, the local statistics of Coulomb gas constrained to its droplet $S$ have been studied. When the localization takes place in a set $\calD$, a hard wall is created at the boundary of $\calD$. If the set $\calD$ contains a neighborhood of droplet $S$, then the localization does not affect the limiting behavior of particles when $N \to \infty$ so much. However, if the hard wall meets the droplet, then the local statistic of the particles at the hard wall changes noticeably. 

Back to the radially symmetric case, fix a radially symmetric potential $Q$ with the droplet 
$S = \{\zeta\in \C : \rho_0 \leq |\zeta| \leq \rho_1\}.$ We assume that $Q$ is strictly subharmonic in a neighborhood of $\d \calD = \{|\zeta| = \rho_{*}\}$, which implies that the function $rq'(r)$ is strictly increasing near $r=\rho_{*}$. 

Assume that $\rho_0 < \rho_{*} \leq \rho_1$. In this case, the equilibrium measure $\sigma^{\calD}$ associated with $Q^{\calD}$ should be supported in the ring 
$$S^{\calD} = \{\zeta\in \C : \rho_0 \leq |\zeta| \leq \rho_{*}\}.$$
In the case when $\rho_0 <\rho_{*} < \rho_1$, 
considering the balayage problem finding a measure sweeping out the measure $\sigma$ in \eqref{Oreq} from $S \setminus \calD$ to $\d \calD$, we have \begin{equation}
d\sigma^{\calD} = \Lap Q\, \1_{S^{\calD}} \,dA + \frac{1-\tau_{*}}{\rho_{*}}\, ds,
\end{equation}
where $\tau_{*} = \frac{1}{2}\rho_{*}q'(\rho_{*})$ and $ds$ is the arclength measure on $\d \calD$ divided by $2\pi$. For the details of the Balayage measure, see \cite[Section 2.4]{ST}. We also refer to \cite[Section 4.1]{CFLV} for the detailed argument for finding the equilibrium measure. 



\subsection{Scaling limits of Coulomb particle systems}

Let $\calD = D(0,\rho_{*})$ with $\rho_0<\rho_{*} <\rho_1$ and fix a number $\alpha > -1$ and a real-valued function $h$ which is radially symmetric and smooth.
We define an $N$-dependent potential $Q_N$ by adding a small logarithmic singularity and a perturbation $h$ to the original potential $Q$,
$$Q_N(\zeta) = Q(\zeta) - \frac{\alpha}{N}\log(\rho_{*} - |\zeta|) - \frac{1}{N}h(\zeta),\quad \zeta \in \calD.$$
This choice of the potential gives a generalization of the potential $-\frac{\alpha}{N}\log(1-|\zeta|^2)$ localized to the unit disk $|\zeta|\leq 1$, which appears in the study of truncated unitary matrices. For simplicity, we assume that $h(\rho_{*})=0$. 

Consider the Coulomb gas system $\Phi^{\calD}=\{\zeta_j\}_1^N$ associated with the potential $Q_N^{\calD}$. As a point process, the distribution of the system is described by correlation functions. For an integer $k$ with $1\leq k \leq N$, the $k$-point correlation function $\bfR_{N,k}$ is defined as
$$\bfR_{N,k}(\zeta_1,\cdots,\zeta_k) = \frac{N!}{(N-k)!} \int_{\C^{N-k}} \Prob_{N,\beta}(\zeta_1,\cdots,\zeta_N)\, dA(\zeta_{k+1})\cdots dA(\zeta_{N}),$$
where $\Prob_{N,\beta}$ is the Boltzmann-Gibbs distribution defined in \eqref{Gbs} associated with the external potential $Q_{N}^{\calD}$.
It is well known that when $\beta=1$, the Coulomb gas system is determinantal, i.e., the $k$-point correlation function is expressed by the determinant 
$$\bfR_{N,k}(\zeta_1,\cdots, \zeta_k) = \det \left(\bfK_{N}(\zeta_i,\zeta_j)\right)_{i,j=1}^k$$
where $\bfK_{N}$ is a function on $\C^2$ called a correlation kernel of the process. Here $\bfK_{N}$ can be taken as the reproducing kernel of a space of weighted polynomials. More precisely, it is expressed as the sum of weighted orthonormal polynomials 
\begin{equation}\label{bfK}
\bfK_{N}(\zeta,\eta) = \sum_{j=0}^{N-1}p_{N,j}(\zeta)\,\bar{p}_{N,j}(\eta)\, e^{-N(Q_{N}^{\calD}(\zeta)+Q_N^{\calD}(\eta))/2}.
\end{equation}
Here $p_{N,j}$ is an orthonormal polynomial of degree $j$ with respect to the measure $e^{-NQ_N^{\calD}}dA$.
In this paper, we study the case $\beta=1$, the determinantal case.

For the system $\Phi^{\calD}=\{\zeta_j\}_1^N$, we define a rescaled system $\{z_j\}_1^N$ at a point on the hard wall $\d \calD$ when $\rho_0<\rho_{*}<\rho_1$. Fix a point $p\in \d \calD$ and rescale the system by setting 
$$\zeta_j = p + \normal
\frac{\rho_{*}\, z_j}{N(1-\tau_{*})},\quad j=1,\cdots, N,$$
where $\normal$ is the unit inward normal to the boundary $\d \calD$ at $p$. 
Write  
$$\gamma_N = \frac{\rho_{*}}{N(1-\tau_{*})}.$$ 
Note that the scaling factor $\gamma_N$ is a radius of the disk at $p$ such that $$N\cdot \sigma_{\calD} \left(D(p,\gamma_N)\right) = C(1+ o(1)), \quad N\to \infty$$
where $C$ is a positive constant and with this scale, a nontrivial scaling limit of the system is obtained. We may assume that $p = \rho_{*} \in \d\calD \cap \R_{+}$ and the normal is in the negative real direction. The rescaled process $\{z_j\}_1^N$ is also determinantal with correlation kernel 
$$K_{N}(z,w) = \gamma_N^2 \bfK_{N}(\rho_{*}-\gamma_N z, \rho_{*}-\gamma_N w).$$ We write $R_{N,k}$, $K_N$ for the $k$-point correlation function and the correlation kernel of the process $\{z_j\}_1^N$, respectively. 
We are now ready to state one of our main results. Let $\RH$ denote the right half plane $\{z\in\C : \re z >0\}$.

\begin{thm}\label{thm:1}
Let $R_{N,k}$ be the $k$-point correlation function of the rescaled system $\{z_j\}_1^N$. For $z_1,\cdots, z_k \in \C$,
$$\lim_{N\to \infty} R_{N,k}(z_1,\cdots,z_k) = \det \Big( K^{(\alpha)}(z_i,z_k)\Big)_{i,j=1}^{k},$$
where 
\begin{equation}\label{limK}
K^{(\alpha)}(z,w) = (2\re z)^{\frac{\alpha}{2}}(2\re w)^{\frac{\alpha}{2}}\int_{0}^{1} \frac{t^{\alpha+1}}{\Gamma(\alpha+1)}\, e^{-t(z+\bar{w})}  \,dt \cdot \1_{\RH}(z)\,\1_{\RH}(w).\end{equation}
The convergence is uniform for $z_1,\cdots,z_k$ in any compact subset of $\RH$.  
\end{thm}

\begin{rmk*}
If $\rho_{*} = \rho_1$, then the hard wall stands along the boundary of the original droplet $S$ and the localization doesn't change the droplet and the equilibrium measure. In this case, the microscopic scale $\gamma_N$ at a regular boundary point (where $\d S$ is smooth and $\Lap Q$ does not vanish) appropriate to obtain a nontrivial edge scaling limit should be of order $N^{-1/2}$ same as in the free boundary case. See \cite{A18, AKM, AKS}. 
\end{rmk*}

\begin{figure}
\includegraphics[width=.32\textwidth]{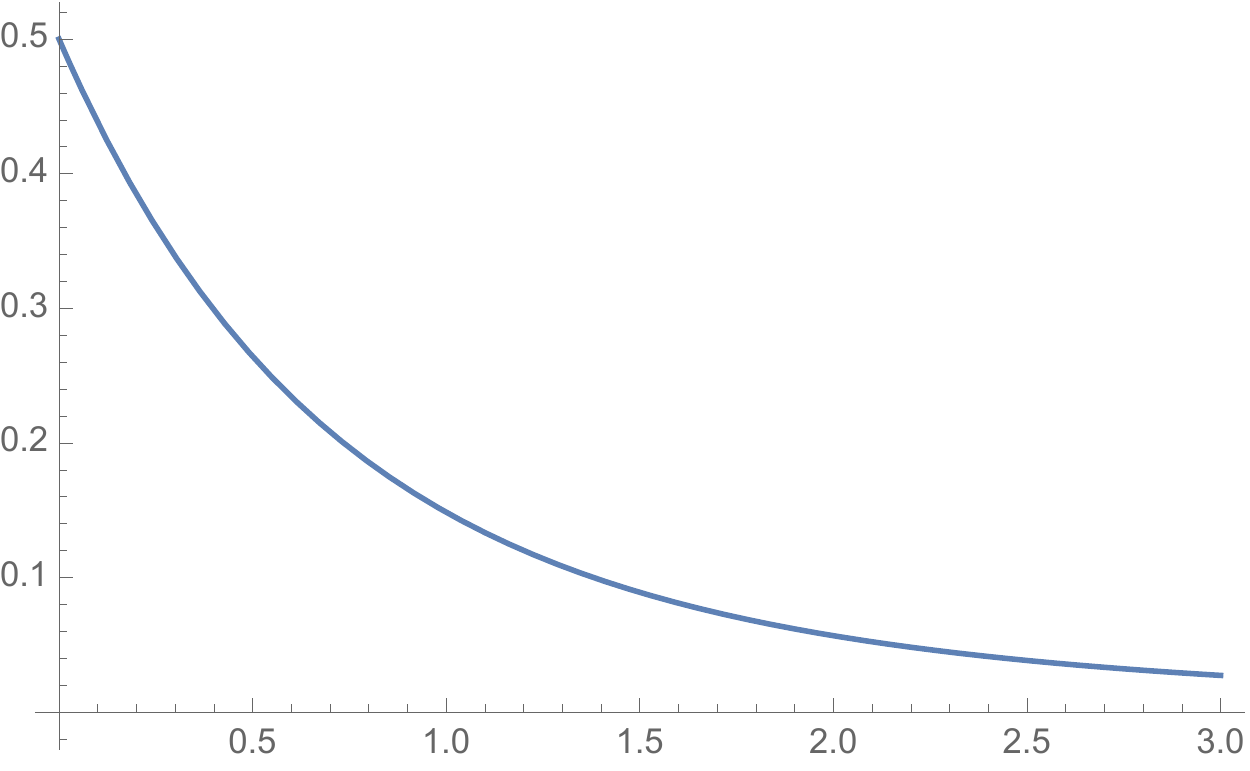}
\includegraphics[width=.32\textwidth]{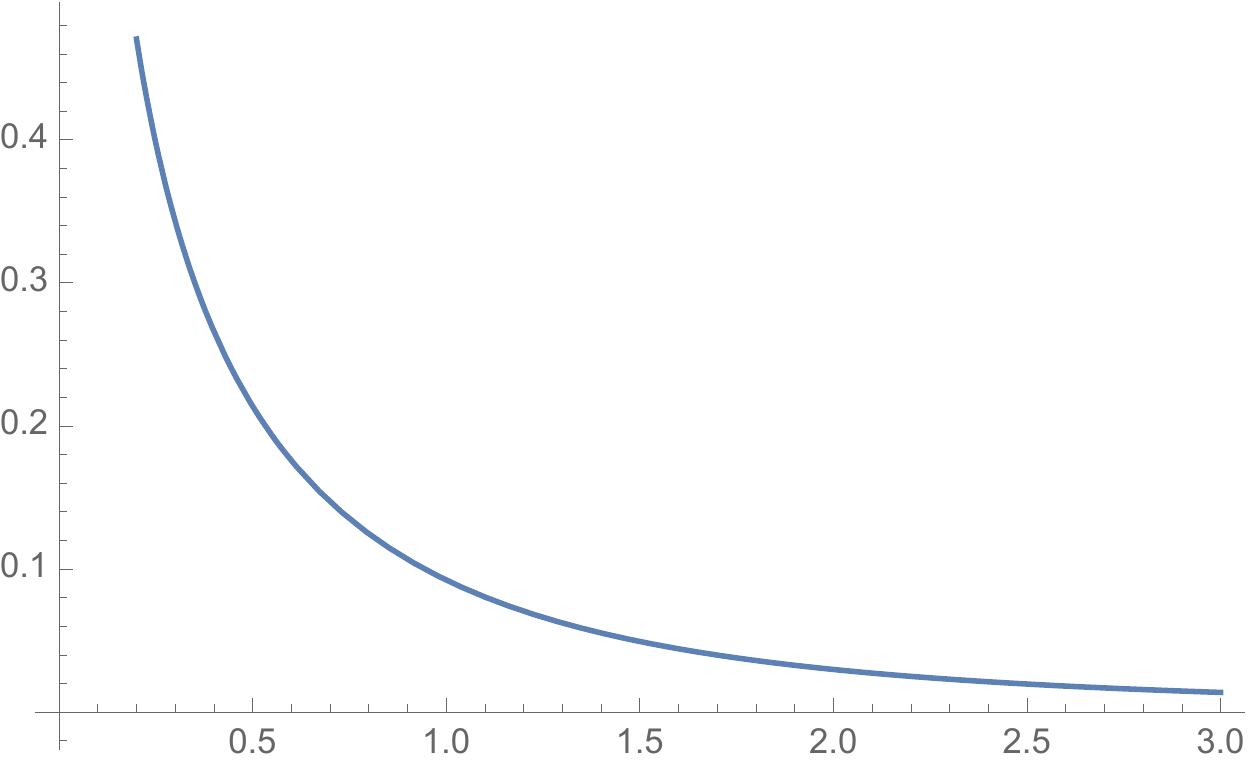}
\includegraphics[width=.32\textwidth]{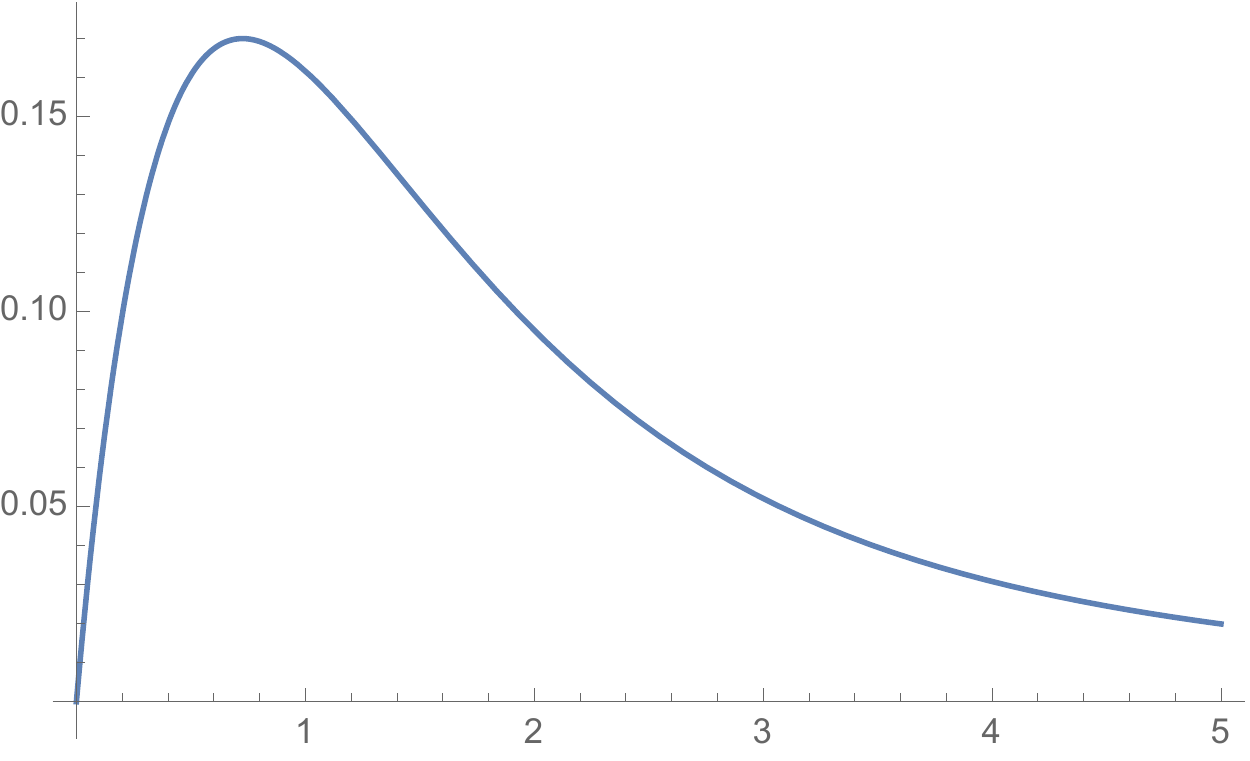}
\caption{Graphs of the limiting $1$-point density $R(x)=K^{(\alpha)}(x,x)$ for $x\in \R_{+}$ with $\alpha=0$, $\alpha=-0.5$, and $\alpha=1$.}\label{fig:alpha}
\end{figure}

Let $\{\zeta_j\}_1^N$ be a random system of Coulomb particles associated with $Q_{N}^{\calD}$. Write $|\zeta|_N$ for the maximal modulus of $\{\zeta_j\}_1^N$, i.e.,
$$|\zeta|_N = \max_{1\leq j\leq N}|\zeta_j|.$$
Rescaling $|\zeta|_N$ about the radius $\rho_{*}$, we define a random variable $\omega_N$ by
$$\omega_N = 2\gamma_N^{-\frac{\alpha+2}{\alpha+1}}\Big(\frac{\rho_{*}}{\Gamma(
\alpha+3)}\Big)^{\frac{1}{\alpha+1}} \left(\rho_{*} - |\zeta|_N \right).$$

\begin{thm}\label{thm:3}
The random variable $\omega_N$ converges in distribution to the Weibull distribution. For $x \geq 0$,
$$\lim_{N\to\infty}\mathbb{P}_N(\omega_N \leq x) = 1- e^{-x^{\alpha+1}}.$$
\end{thm}

The maximal modulus has fluctuation of order $O(N^{-\frac{\alpha+2}{\alpha+1}})$ near the hard wall.
One can see the microscopic effect of the logarithmic singularity $-\frac{\alpha}{N}\log(\rho_{*} - |z|)$ at the hard wall $|\zeta| = \rho_{*}$. (See Figure \ref{fig:alpha}.) We observe a repulsive force between the particles and the hard wall if $\alpha$ is positive and an attractive force between them if $\alpha$ is negative. 

We shall give a proof of Theorem \ref{thm:1} and Theorem \ref{thm:3} using the asymptotics of orthonormal polynomials in Section \ref{sec:hard}.


\subsection{Ward's equation}
We shall discuss an abstract approach to universality based on Ward's equation, a differential equation of $R(z)=K(z,z)$. The rescaled version of Ward's equation considered in this paper does not depend on specific potentials $Q$, and we shall describe solutions $R$ of the equation under some suitable assumptions on $R$. 

Theorem \ref{thm:1} shows the universality of limiting correlation functions for radially symmetric potentials. 
One can observe that the limiting correlation kernel $K=K^{(\alpha)}$ is invariant under the vertical translation,
$$K(z,w) = K(z+is, w+is),\quad s\in \R,$$
i.e., the limiting $1$-point function $R(z)=K(z,z)$ is a function only in $\re z$.

We consider a class of functions defined in $\RH^2$ where
each function is of the form "Laplace-type integral"
\begin{equation}\label{Lapint}
K_f(z,w) = (2\re z)^{\frac{\alpha}{2}}(2\re w)^{\frac{\alpha}{2}}\int_{0}^{\infty} f(t)\, e^{-t(z+\bar{w})}\, dt
\end{equation}
for some Borel measurable function $f$ on $\R_{+}$. Here, if $f(t)=\Gamma(\alpha+1)^{-1}t^{\alpha+1}\cdot \1_{(0,1)}$, then $K_f$ is the limiting correlation kernel $K$ in \eqref{limK}. 

\begin{thm}\label{thm:2}
The limiting correlation kernel $K$ in \eqref{limK} satisfies the mass-one equation $$\int_{\RH} |K(z,w)|^2 \, dA(w) = K(z,z),$$ and the rescaled Ward equation $$\dbar C(z) = R(z) - \Lap \log R(z),\quad z\in \RH,$$
where $R(z)=K(z,z)$ and $C$ is the Cauchy transform
$$C(z)=\int_{\RH} \frac{|K(z,w)|^2}{K(z,z)}\frac{dA(w)}{z-w}.$$
Conversely, if a function $K_f$ of the form \eqref{Lapint}
satisfies the rescaled Ward equation, then there exists a connected set $I$ in $\R_{+}$ such that
$f(t)=\Gamma(\alpha+1)^{-1}t^{\alpha+1}\cdot \1_{I}$ almost everywhere in $\R_{+}$.
\end{thm}

\subsection{Comments and related works}
\subsubsection{Microscopic properties of two-dimensional Coulomb gases}
In the microscopic regime, scaling limits of two-dimensional determinantal Coulomb gases depend on the regularity of the equilibrium density and also that of the boundary of the droplet. Microscopic properties near a bulk singularity, an isolated point in the interior of the droplet at which the equilibrium density $\Lap Q$ vanishes, are determined by the dominant terms in the Taylor expansion of $\Lap Q$ at the point \cite{AS1}.
In \cite{AKS0}, scaling limits of the correlation functions near a logarithmic singularity in the bulk of the droplet have been studied and the universality of the limiting correlation kernel expressed by Mittag-Leffler functions has been proved. For the edge scaling limit near a regular boundary point at which $\Lap Q>0$ and $\d S$ is smooth, the universality of $\erfc$ kernel has been studied in \cite{AKM, HW17}. Scaling limits near a singular boundary point which is a cusp or a double point have been investigated in \cite{AKMW}.

\subsubsection{The maximal modulus of Coulomb gases}
In the free boundary case, it has been proved that a properly scaled maximal modulus of Coulomb gases follows Gumbel distribution for $Q(\zeta)=|\zeta|^2$ (Complex Ginibre ensemble) \cite{R03} and for a class of radially symmetric potentials \cite{CP14}. 
In the hard edge case, the case when the Coulomb particles are confined to the droplet $S$, i.e., $\rho_{*} = \rho_1$ in our setting, the distribution of the properly scaled maximal modulus converges to the exponential distribution when $\alpha=0$. See \cite{Seo}. In this case, the scaling factor is of order $N^{-1}$ while the proper one for the case when $\rho_{*}<\rho_1$ is of order $N^{-2}$. For the potential $Q(\zeta) = -\frac{\alpha}{N}\log(1-|\zeta|^2)$, the convergence of the distribution of the maximal modulus to the Weibull distribution has been studied in \cite{GQ, LGMS}.

\subsubsection{Coulomb gases with a hard wall}
A family of two-dimensional determinantal Coulomb gases confined to an ellipse has been investigated in \cite{NAKP}. This model gives an elliptic generalization of truncated unitary matrix ensembles and also new universality classes in the limit of weak and strong non-Hermiticity. In \cite{GZ18}, the fluctuation of the extremal Coulomb particle has been studied in the case when the equilibrium measure is the uniform measure on the circle and when the hard edge constraint is imposed on the unit circle. For Coulomb gases in any dimension,
the third-order transition between the `pushed' and the `pulled' phases, which correspond to the cases when $\rho_{*} < \rho_1$ and when $\rho_{*} > \rho_1$ respectively in our setting,
has been studied in \cite{CFLV}.

\section{Hard edge scaling limits}\label{sec:hard}

For a number $\alpha > -1$ and a smooth real-valued function $h$ which is radially symmetric and $h(\rho_{*})=0$, we consider the potential $$Q_N(\zeta)=Q(\zeta) - \frac{\alpha}{N}\log(\rho_{*}-|\zeta|)-\frac{1}{N}h(\zeta),\quad |\zeta|<\rho_{*}.$$
Since $Q$ is radially symmetric, the orthonormal polynomials $p_{N,j}$ can be taken as monomials. 
For the Coulomb gas system $\Phi^{\calD}$ associated with the localized potential $$Q^{\calD}_N(\zeta) = Q_N(\zeta)\cdot \1_{\calD}(\zeta) + \infty\cdot \1_{\C\setminus\calD}(\zeta),\quad \calD = D(0,\rho_{*})$$ defined in Section \ref{sec:loc}, the corresponding correlation kernel is given by \begin{align}\label{corrk}
\bfK_N(\zeta,\eta) = \sum_{j=0}^{N-1}\frac{(\zeta\bar{\eta})^j}{\|\zeta^j\|^2_{L^2(\mu_N)}}e^{-N(Q_{N}(\zeta)+Q_{N}(\eta))/2}\1_{\calD}(\zeta) \1_{\calD}(\eta),\end{align}
where $d\mu_N = e^{-N Q_{N}}\1_{\calD}dA$. Here, we write $p_{N,j}$ for the orthonormal polynomial of degree $j$ with respect to the measure $\mu_N$, i.e.,
$$p_{N,j}(\zeta) = \zeta^j /\|\zeta^j\|_{L^2(\mu_N)}.$$

\subsection{Asymptotics of rescaled correlation kernels.} \label{ssec:asym}

Let $\tau_{*}$ denote the number $\frac{1}{2}\rho_{*} \, q'(\rho_{*})$ and $m_N$ denote the number $N\tau_{*} + M\sqrt{N}$ for sufficiently large $M$.
We first consider the orthonormal polynomial of degree $j$ with $m_N \leq j \leq N-1$ for sufficiently large $M$. These terms of higher degree 
are dominant in the sum \eqref{corrk} near the boundary $\d \calD$. Write
\begin{equation}\label{Kshp}
K_N^{\#}(z,w)= \gamma_N^2 \sum_{j \geq m_N}^{N-1} 
p_{N,j}(\zeta)\bar{p}_{N,j}({\eta})\,
e^{-N(Q_{N}(\zeta)+Q_{N}(\eta))/2}\end{equation}
where $\zeta = \rho_{*} - \gamma_N z$, $\eta=\rho_{*} - \gamma_N w$. We write $\tau = j/N$.

To compute the limit of $K_N^{\#}$ as $N\to \infty$, we first consider the function 
$$V_{\tau}(r)=q(r) - 2\tau\log r,\quad r>0$$
and observe that the Taylor expansion
\begin{equation}\label{ExV}
V_{\tau}(r) = V_{\tau}(\rho_{*}) + \frac{2}{\rho_{*}}\left(\tau_{*}-\tau \right)(r-\rho_{*}) + O(r-\rho_{*})^2,\quad r\to \rho_{*},  
\end{equation}
holds by a simple calculation
$$V_{\tau}'(\rho_{*}) = q'(\rho_{*})-\frac{2\tau}{\rho_{*}} = \frac{1}{\rho_{*}}\left(\rho_{*}\,q'(\rho_{*}) - {2\tau} \right)=\frac{2}{\rho_{*}}\left(\tau_{*} - \tau\right).$$


\begin{lem}\label{lem:nm1}
Assume that $j$ is in the range $m_N \leq j \leq N-1$. We have
$$e^{NV_{\tau}(\rho_{*})}\, \|\zeta^j\|_{L^2(\mu_N)}^2  = 2\rho_{*} \,\Gamma(\alpha+1)  
\Big(\frac{\rho_{*}}{2N\left(\tau-\tau_{*}\right)}\Big)^{\alpha+1}(1+o(1))$$ as $N\to \infty.$
Here, the error term $o(1)$ is uniform in $j$.
\end{lem}

\begin{proof}
For fixed $j$, we see that
\begin{equation}\label{nmint}
\|\zeta^j\|^{2}_{L^2(\mu_N)} 
= \int_{\rho_{*} -\delta_N}^{\rho_{*}} 2r^{2j+1}\,e^{-NQ_N(r)}\,dr + \int_{0}^{\rho_{*}-\delta_N}2r^{2j+1}\,e^{-NQ_N(r)} \,dr
\end{equation}
where $\delta_N$ denotes a small number $N^{-\frac{1}{2}}\log N$.
By \eqref{ExV}, the first integral in the right-hand side is calculated as
\begin{align*}
\int_{\rho_{*}-\delta_N}^{\rho_{*}} 2r(\rho_{*}-r)^{\alpha} e^{h(r)} e^{-NV_{\tau}(r)}\,dr 
&= 2\rho_{*} \, e^{-NV_{\tau}(\rho_{*})}\int_{0}^{\delta_N}t^{\alpha}e^{-\frac{2N}{\rho_{*}}\left(\tau-\tau_{*}\right)t}dt \, (1+o(1))\\ 
&= 2\rho_{*} \, \Gamma(\alpha+1) \,  e^{-NV_{\tau}(\rho_{*})} 
\Big(\frac{\rho_{*}}{2N\left(\tau-\tau_{*}\right)}\Big)^{\alpha+1} (1+o(1)),
\end{align*}
where $o(1) \to 0$ uniformly in $j$ as $N\to \infty$.

Now it suffices to show that the second integral in the right-hand side of \eqref{nmint} is negligible. We observe that for $\tau = j/N$ with $m_N \leq j \leq N-1$
$$V_{\tau}'(r) = \frac{1}{r}\left(rq'(r) - 2\tau\right) \leq \frac{1}{r}\Big( rq'(r)  - \frac{2m_N}{N}\Big).$$
Since $Q$ is subharmonic in $\C$, $rq'(r)$ is increasing in $r$ and 
$$V_{\tau}'(r) \leq \frac{1}{r}\Big( \rho_{*}\, q'(\rho_{*}) - 2\Big(\tau_{*} + \frac{M}{\sqrt{N}}\Big)  \Big) = -\frac{2M}{r\sqrt{N}},\quad r\leq \rho_{*}.$$
Thus, we see that for all $r$ with $r\leq \rho_{*}-\delta_N$
$$V_{\tau}(r)\geq V_{\tau}(\rho_{*}-\delta_N)$$
and by Taylor's theorem, we obtain the expansion
$$V_{\tau}(\rho_{*} -\delta_N) = V_{\tau}(\rho_{*}) - V_{\tau}'(r_{*})\,\delta_N$$ 
for some $r_{*} \in [\rho_{*}-\delta_N,\rho_{*}]$.
Hence, there exists a constant $c>0$ such that for all $j$
\begin{align*}
e^{NV_{\tau}(\rho_{*})}\int_{0}^{\rho_{*}-\delta_N} 2r\,(\rho_{*}-r)^{\alpha}e^{h(r)} e^{-NV_{\tau}(r)} dr 
&\leq \frac{2\rho_{*}^{\alpha+2}}{\alpha+1}\|e^h\|_{L^{\infty}} e^{-N(V_{\tau}(\rho_{*}-\delta_N)-V_{\tau}(\rho_{*}))}\\
&\leq C \,e^{-\frac{2M}{r_{*}}\log N}=O(N^{-c M}).
\end{align*} 
This completes the proof.
\end{proof}

A function $c(\zeta,\eta)$ is called a cocycle if $c(\zeta,\eta) = g(\zeta)\bar{g}(\eta)$ for a continuous unimodular function $g$. Note that a correlation kernel is only determined up to multiplication by a cocycle.

\begin{lem} \label{lem:ksp}
There exists a cocycle $c(z,w)$ such that 
\begin{equation*} 
K_N^{\#}(z,w)= c(z,w)(2\re z)^{\frac{\alpha}{2}}(2\re w)^{\frac{\alpha}{2}} \int_{0}^{1} \frac{t^{\alpha+1}}{\Gamma(\alpha+1)} e^{-t(z+\bar{w})} dt \cdot \1_{\RH}(z)\1_{\RH}(w) \,(1+o(1)),
\end{equation*}
where $o(1)\to 0$ uniformly in any compact subset of $\RH^2$ as $N\to \infty$.
\end{lem}
\begin{proof}
Let $\calK$ be a compact subset of $\RH$. For all $\tau=j/N$ with $m_N \leq j \leq N-1$, we have
\begin{align}\label{exQ}
Q(\zeta)-2\tau\log \zeta 
&= V_{\tau}(\rho_{*})-\frac{2\left(\tau_{*}-\tau \right)}{N(1-\tau_{*})}\re z + \frac{2\tau}{N(1-\tau_{*})} i\im z + O(N^{-2})
\end{align}
for $\zeta = \rho_{*} - \gamma_N z$,
where the $O$-constant is uniform for $z \in \calK$. Also we observe that 
\begin{align*}
(\rho_{*} - |\rho_{*} - \gamma_N z|)^{\alpha} e^{h(\zeta)} = (\gamma_N \re z)^{\alpha}(1+o(1)),\quad \zeta = \rho_{*} - \gamma_N z.
\end{align*}
For each $\tau$ and $z$, we write 
$$f_{\tau}(z) = (\tau-\tau_{*})^{\frac{\alpha+1}{2}} \exp\Big( - \frac{\tau-\tau_{*}}{1-\tau_{*}} \,z \Big).$$
Combining with Lemma \ref{lem:nm1}, we obtain
\begin{align*}
\frac{\zeta^j \, e^{-\frac{N}{2}Q_{N}(\zeta)}}{\|\zeta^j\|_{L^2(\mu_N)}}\, 
= \frac{1}{\rho_{*}} \Big(\frac{N}{\Gamma(\alpha+1)}\Big)^{\frac{1}{2}} \frac{(2\re z)^{\frac{\alpha}{2}}}{(1-\tau_{*})^{\frac{\alpha}{2}}} e^{-\frac{\tau_{*}}{1-\tau_{*}}i \im z} f_\tau(z) \,\1_{\calH_N}(z)\,(1+o(1))
\end{align*}
where 
$\calH_N=\{z\in\C: \rho_{*} -\gamma_N z \in \calD \}$,
and $o(1)\to 0$ uniformly for $z\in\calK$ and $j$ with $m_N \leq j \leq N-1$ as $N\to \infty$.
Note that $\calH_N$ converges to $\RH$ as $N\to \infty$.
Write $$c(z,w) = \exp\Big(-\frac{\tau_{*}}{1-\tau_{*}} i \im (z-w)\Big),$$ which is a cocycle. 
Using the Riemann sum approximation with step length $N^{-1}$, we have
\begin{align*}
K_N^{\#}(z,w) 
&=\gamma_N^2 \sum_{j \geq m_N}^{N-1} \frac{(\zeta\bar{\eta})^j \, e^{-\frac{N}{2}(Q_N(\zeta)+Q_N(\eta))}}{\|\zeta^j\|_{L^2(\mu_N)}^2} \cdot \1_{\calD}(\zeta)\1_{\calD}(\eta)\\
&= \frac{(2\re z)^{\frac{\alpha}{2}}(2\re w)^{\frac{\alpha}{2}}}{\Gamma(\alpha+1)(1-\tau_{*})^{\alpha+2}} \,c(z,w)\, \1_{\RH}(z)\1_{\RH}(w)\, N^{-1}\!\!\sum_{j\geq m_N}^{N-1} f_{\tau}(z) f_{\tau}(\bar{w})\,(1+o(1))\\
&=\frac{(2\re z)^{\frac{\alpha}{2}}(2\re w)^{\frac{\alpha}{2}}}{\Gamma(\alpha+1)(1-\tau_{*})^{\alpha+2}} \,c(z,w)\, \1_{\RH}(z)\1_{\RH}(w) \int_{0}^{1-\tau_{*}} \!\!\!s^{\alpha+1} e^{-\frac{s}{1-\tau_{*}}(z+\bar{w})} ds\,(1+o(1)), 
\end{align*}
where $o(1)\to 0$ locally uniformly in $\RH^2$ as $N\to \infty$. By the change of variable $s=(1-\tau_{*})t$, we obtain the convergence
$$K_N^{\#}(z,w) = \frac{(2\re z)^{\frac{\alpha}{2}}(2\re w)^{\frac{\alpha}{2}}}{\Gamma(\alpha+1)} \,c(z,w)\, \1_{\RH}(z)\1_{\RH}(w)\int_{0}^{1} t^{\alpha+1}\, e^{-t(z+\bar{w})}dt\, (1+o(1))$$
as $N\to \infty$.
\end{proof}

\subsection{Discarding the lower degree polynomials}\label{ssec:l}
This subsection provides some estimates for the weighted orthonormal polynomials of degree $j$ with $0\leq j< m_N$, and we shall prove that the weighted polynomials of lower degree can be neglected in the sum \eqref{corrk}. 

First we consider the case when $N(\tau_{*}-\delta_N)\leq j\leq m_N$. This means that for $\tau = j/N$,
\begin{equation*}\label{tau_r}
\tau_{*}-\delta_N \leq \tau \leq \tau_{*} + M/\sqrt{N}.
\end{equation*}
Recall that $\rho_{*}q'(\rho_{*})=2\tau_{*}$ and $rq'(r)$ is increasing. Since $Q$ is strictly subharmonic near $|\zeta|=\rho_{*}$, for each $\tau$ in this region \eqref{tau_r} there exists $\rho_\tau$ in a neighborhood of $\rho_{*}$ such that $\rho_{\tau}q'(\rho_\tau)=2\tau$ and $\Lap Q(\rho_\tau)>\epsilon$ for some positive number $\epsilon$ which is independent of $\tau$.
Then we observe that  $V_{\tau}'(\rho_\tau)=0$ and $V_{\tau}''(\rho_{\tau})=4\Lap Q(\rho_{\tau})$ and obtain that   
\begin{equation}\label{crtexp}
V_{\tau}(r) = V_{\tau}(\rho_\tau) + {2\Lap Q(\rho_\tau)}(r-\rho_\tau)^2 +O(|r-\rho_\tau|)^3, \quad r\to \rho_\tau.
\end{equation}  
We also observe that the following asymptotic expansion holds:
$$2(\tau_{*}-\tau)=\rho_{*} q'(\rho_{*})- \rho_\tau q'(\rho_\tau)=4\rho_{*} \Lap Q(\rho_{*})(\rho_{*}-\rho_\tau)+O(|\rho_{*}-\rho_\tau|)^2,\quad \rho_\tau \to \rho_{*},$$ 
which implies the asymptotic relation
\begin{equation}\label{asymrtau}
\rho_{*}-\rho_{\tau} = \frac{\tau_{*}-\tau}{2\rho_{*} \Lap Q(\rho_{*})}+O(|\tau_{*}-\tau|)^2,\quad \tau \to \tau_{*}.
\end{equation}
The following lemma gives an estimate of the norm $\|\zeta^j\|_{L^2(\mu_N)}$. For the statement, we write 
\begin{equation}\label{phixi}
\varphi_{\alpha}(\xi) = \int_{0}^{\infty} s^{\alpha}e^{-\frac{1}{2}(s-\xi)^2}\,ds\quad \mbox{and}\quad \xi_{\tau} = \frac{\sqrt{N}(\tau_{*}-\tau)}{\rho_{*}\sqrt{\Lap Q(\rho_{*})}}.
\end{equation}

\begin{lem} \label{lem:nm2}
Fix $j$ with $N(\tau_{*}-\delta_N)\leq j\leq m_N$.  
Then we obtain
$$\|\zeta^j\|^2_{L^2(\mu_N)} = C N^{-\frac{\alpha+1}{2}} e^{-NV_{\tau}(\rho_\tau)}\varphi_{\alpha}(\xi_{\tau})(1+o(1)),\quad N\to \infty,$$
where the constant $C$ is given by $C= 2\rho_{*} (4\Lap Q(\rho_{*}))^{-\frac{\alpha+1}{2}}$.
\end{lem}

\begin{proof}
For all $j$ with $N(\tau_{*}-\delta_N)\leq j \leq m_N$, 
\begin{equation}\label{rhorho}
-c_1 M\cdot N^{-1/2} \leq \rho_{*} - \rho_{\tau} \leq c_2 \,\delta_N
\end{equation} 
where $c_1$ and $c_2$ are positive constants by \eqref{asymrtau}.
Similar to the proof of Lemma \ref{lem:nm1}, we derive the estimate
\begin{equation}\label{eq:nm2}
\|\zeta^j\|^2_{L^2(\mu_N)} = 2\rho_{*} \,e^{-NV_{\tau}(\rho_\tau)} \int_{\rho_{\tau}-\delta_N}^{\rho_{*}} (\rho_{*} - r)^{\alpha} e^{-2N\Lap Q(\rho_{*})(r-\rho_\tau)^2}dr\,(1+o(1))\end{equation}
as $N\to \infty$
by using the expansion \eqref{crtexp}.  
By changing the variable $s=\sqrt{4N\Lap Q(\rho_{*})}(r-\rho_\tau)$ in \eqref{eq:nm2}, we obtain
\begin{align*}
\|\zeta^j\|^2_{L^2(\mu_N)} = 2\rho_{*} \, e^{-NV_{\tau}(\rho_\tau)}(4N\Lap Q(\rho_{*}))^{-\frac{\alpha+1}{2}}\varphi_{\alpha}(\xi_{\tau}) \,(1+o(1)),\quad N\to \infty,
\end{align*}
which completes the proof.
\end{proof}

\begin{lem} \label{lem:sum2}
Let $\calK$ be a compact subset of $\RH$. We obtain for $\zeta=\rho_{*} - \gamma_N z$
\begin{equation*}
\gamma_N^2\sum_{N(\tau_{*}-\delta_N)}^{m_N} |p_{N,j}(\zeta)|^2 e^{-NQ_N^{\calD}(\zeta)}=O(N^{-(1+\frac{\alpha}{2})}\log N), \quad N\to \infty,
\end{equation*}
where the $O$-constant is uniform for $z\in \calK$.
\end{lem}

\begin{proof}
Observe that for each $j$ 
\begin{align*}
|p_{N,j}(\zeta)|^2 e^{-NQ_N^{\calD}(\zeta)} 
=\frac{N^{\frac{\alpha+1}{2}}}{C\varphi_{\alpha}(\xi_{\tau})}|\zeta|^{2j} e^{-N (Q_N(\zeta) -V_{\tau}(\rho_\tau))}\cdot \1_{\calD}(\zeta)\,(1+o(1))
\end{align*}
by Lemma \ref{lem:nm2}. Here, $\varphi_{\alpha}$ and $\xi_\tau$ are defined in \eqref{phixi} before Lemma \ref{lem:nm2}.
Using the expansions \eqref{crtexp} and \eqref{asymrtau}, we obtain
\begin{align*}
|p_{N,j}(\zeta)|^2 e^{-NQ_N^{\calD}(\zeta)}
&=\frac{N^{\frac{\alpha+1}{2}}(\gamma_N \re z)^{\alpha}}{C\varphi_{\alpha}(\xi_{\tau})} e^{-2N\Lap Q(\rho_\tau)(\rho_{*}-\rho_\tau - \gamma_N \re z)^2}\cdot\1_{\RH}(z)\,(1+o(1))\\
&=C'N^{\frac{1-\alpha}{2}} (\re z)^{\alpha} \frac{e^{-\frac{1}{2}\xi_{\tau}^2}}{\varphi_{\alpha}(\xi_{\tau})}\cdot\1_{\RH}(z)\,(1+o(1))
\end{align*}
for some constants $C$ and $C'$ which are independent of $j$.
Here 
$o(1)\to 0$ uniformly for $z\in \calK$ as $N\to \infty$. 
We obtain a bound 
$$\frac{e^{-\frac{1}{2}{\xi_{\tau}^2}}}{\varphi_{\alpha}(\xi_{\tau})}\leq C_1 M^{\alpha+1}$$
for some constant $C_1$. This implies that there exists a constant $C_2>0$ such that for all $z\in \calK$ and sufficiently large $N$
\begin{align*}
\gamma_N^2 \sum_{N(\tau_{*}-\delta_N)}^{m_N} |p_{N,j}(\zeta)|^2 e^{-NQ_N^{\calD}(\zeta)}\leq C_2 N^{-\frac{\alpha+1}{2}}\delta_N, 
\end{align*}
which completes the proof.
\end{proof}

For the lower degree polynomials with $j \leq N(\tau_{*}-\delta_N)$, the localization of the potential $Q$ to $\calD = D(0,\rho_{*})$ does not affect the asymptotics of weighted polynomials that much because the weighted orthogonal polynomial $|p_{N,j}|^2 e^{-NQ_N}$ tends to decay exponentially outside the set $D(0,\rho_{\tau})
$ which is at least $O(\delta_N)$-distance away from the hard wall $\d \calD$. The sum of lower degree terms with $j \leq N(\tau_{*}-\delta_N)$ can be estimated as follows:


\begin{lem}\label{lem:sum3}
Let $\calK$ be a compact subset of $\RH$. We obtain for $\zeta = \rho_{*} - \gamma_N z$
$$\gamma_N^2 \sum_{j=0}^{N(\tau_{*}-\delta_N)} 
|p_{N,j}(\zeta)|^2 e^{-NQ_{N}^{\calD}(\zeta)}
= O(e^{-c(\log N)^2}),\quad N\to \infty,$$
where $c$ is a positive constant and the $O$-constant is uniform for $z\in\calK$.
\end{lem}

\begin{proof}

Write $\tau_{*}^{(0)}=\tau_{*}-\delta_N$ and $\tau_{*}^{(1)}=\tau_{*}-\frac{1}{2}\delta_N$. We also write $\rho_{*}^{(0)}=\rho_{\tau_{*}^{(0)}}$ and 
$\rho_{*}^{(1)}=\rho_{\tau_{*}^{(1)}}$ for simplicity.
We first observe that by \eqref{asymrtau} there exists a constant $c_1>0$ such that
$\rho_{*} - \rho_{*}^{(1)} > c_1\delta_N$ and $\rho_{*}^{(1)}-\rho_{*}^{(0)}> c_1\delta_N$ simultaneously. 
Since for each $\tau \leq \tau_{*}^{(0)}$ it holds that 
$V_{\tau}(r) \leq V_{\tau}(\rho_{*}^{(1)})$
for all $r$ with $\rho_{*}^{(0)}\leq r \leq \rho_{*}^{(1)}$, we obtain that for $j$ with $j\leq N\tau_{*}^{(0)}$
\begin{align*}
\|\zeta^j\|^2_{L^2(\mu_N)}\geq \int_{\rho_{*}^{(0)}}^{\rho_{*}^{(1)}} 2r(\rho_{*} - r)^{\alpha} e^{h(r)} e^{-NV_{\tau}(r)} dr \geq  C_1 \delta_N^{\alpha+1} e^{-NV_{\tau}(\rho_{*}^{(1)})}
\end{align*} 
for some constant $C_1>0$. Observing that for $r=|\rho_{*}-\gamma_N z|$ with $z\in \calK$
\begin{align*}
V_{\tau}(r)- V_{\tau}(\rho_{*}^{(1)}) 
&= V_{\tau_{*}^{(1)}}(r)-V_{\tau_{*}^{(1)}}(\rho_{*}^{(1)})+2(\tau-\tau_{*})\log \frac{\rho_{*}^{(1)}}{r}\\
&\geq 2\Lap Q(\rho_{*}^{(1)}) (r-\rho_{*}^{(1)})^2 \geq c_2 \,\delta_N^2
\end{align*}
for some constant $c_2>0$, we have for all $j$ with $j\leq N\tau_{*}^{(0)}$
$$e^{-N(V_{\tau}(r)- V_{\tau}(\rho_{*}^{(1)}))} = O(e^{-c_2(\log N)^2})$$
as $N\to \infty$. This implies that for $j\leq N\tau_{*}^{(0)}$
\begin{align*}
|p_{N,j}(\zeta)|^2 e^{-NQ_N(\zeta)} = e^{h(\zeta)}(\rho_{*} - |\zeta|)^{\alpha}\frac{e^{-NV_{\tau}(|\zeta|)}}{\|\zeta^j\|^2_{L^2(\mu_N)}}  \leq C_2 \gamma_N^{\alpha} \delta_N^{-(\alpha+1)} e^{-c_2(\log N)^2}.
\end{align*}
Thus, we conclude that there exists a constant $c$ such that as $N\to\infty$,
\begin{equation*}
\gamma_N^2 \sum_{j= 0}^{N(\tau_{*}-\delta_N)}|p_{N,j}(\zeta)|^2e^{-NQ_N(\zeta)} = O(e^{-c(\log N)^2}),
\end{equation*}
which completes the proof.
\end{proof}

Combining Lemma \ref{lem:sum2} and Lemma \ref{lem:sum3},
we obtain that for $\zeta = \rho_{*} - \gamma_N z$ and $\eta = \rho_{*} -\gamma_N w$
\begin{align*}\label{erpart}
\gamma_N^2 \sum_{j=0}^{m_N}p_{N,j}(\zeta)\bar{p}_{N,j}(\eta) e^{-N(Q_{N}(\zeta)+Q_{N}(\eta))/2} = O(N^{-(1+\frac{\alpha}{2})}(\log N)^2),\quad N\to \infty
\end{align*}
by the Cauchy-Schwarz inequality. Here $O$-constant is uniform in compact subsets of $\RH^2$. 
Hence, the rescaled correlation kernel can be written as
\begin{equation}
K_N(z,w) = K_N^{\#}(z,w) + o(1),
\end{equation} 
where $o(1)\to 0$ as $N\to \infty$ locally uniformly in $\RH^2$. Now Lemma \ref{lem:ksp} completes the proof of Theorem \ref{thm:1}.

\subsection{Maximal modulus}
In this subsection, we prove Theorem \ref{thm:3}. To analyze the distribution of the maximal modulus $|\zeta|_N = \max_{1\leq j\leq N} |\zeta_j|$, we first observe that the distribution function of $|\zeta|_N$ is represented by the gap probability that there is no particle outside a disk, i.e.,
$$\mathbb{P}_N(|\zeta|_N \leq r) = \frac{1}{Z_N}\int_{D(0,r)^N}\Prob_{N,1}(\zeta_1,\cdots,\zeta_N)\,dA(\zeta_1)\cdots dA(\zeta_N).$$
Since $Q_N$ is radially symmetric, the determinantal structure of the system gives that
\begin{equation}\label{pbmamo}
\mathbb{P}_N(|\zeta|_N \leq r) = \prod_{j=0}^{N-1}  \left(1 - \int_{\calD \setminus D(0,r)} |p_{N,j}|^2 e^{-NQ_N} dA \right)
\end{equation}
for $r\leq \rho_{*}$. See \cite[Section 15.1]{Meh} and \cite{Dei} for the computation of gap probabilities. Also note that it was observed in \cite{Kos} that in the case when $Q(\zeta) = |\zeta|^2$, the maximal modulus of the particles has the same distribution as that of independent random variables.


Write $$a_N = \frac{1}{2}\,\gamma_N^{\frac{\alpha+2}{\alpha+1}}\Big(\frac{\rho_{*}}{\Gamma(\alpha+3)}\Big)^{-\frac{1}{\alpha+1}},\quad \gamma_N = \frac{\rho_{*}}{N(1-\tau_{*})}.$$
The rescaled maximal modulus $$\omega_N = a_N^{-1}\cdot(\rho_{*} -|\zeta|_N )$$ has a distribution
\begin{equation*}
\mathbb{P}_N\left(\omega_N \geq x\right) = \mathbb{P}_N\left(|\zeta|_N \leq \rho_{*} - a_N \,x\right).
\end{equation*}
We consider the sum
\begin{equation}\label{sumIN}
I_N(x):=\sum_{j=0}^{N-1}\int_{\calD\setminus D(0,\rho_N(x))}|p_{N,j}|^2 e^{-NQ_N} dA, \quad \rho_N(x) = \rho_{*} - a_N\, x,\end{equation}
and obtain the following lemma.
\begin{lem} \label{lem:IN}
We have that for $x\geq 0$,
$$\lim_{N\to\infty} I_N(x) = x^{\alpha+1}.$$
Here the convergence is uniform in $\R_{+}$.
\end{lem}

\begin{proof}
As in Section \ref{ssec:asym}, we write the integral in \eqref{sumIN} in terms of the function $V_{\tau}$, $$I_N(x) = \sum_{j=0}^{N-1} \|\zeta^j\|_{L^2(\mu_N)}^{-2}\int_{\rho_N(x)}^{\rho_{*}} 2r(\rho_{*}-r)^{\alpha}e^{h(r)}e^{-NV_{\tau}(r)} \, dr.$$ 
First fix $j$ with $m_N \leq j \leq N-1$ where $m_N =N\tau_{*}+M\sqrt{N}$ for sufficiently large $M$. Let $\calK$ be a compact subset of $\R_{+}$. From the proofs of Lemma \ref{lem:nm1} and Lemma \ref{lem:ksp}, we obtain that for all $x\in \calK$
\begin{align} \label{domint}
&\|\zeta^j\|_{L^2(\mu_N)}^{-2}\int_{\rho_N(x)}^{\rho_*} 2r(\rho_*-r)^{\alpha}e^{h(r)}e^{-NV_{\tau}(r)} \, dr \\
&= \frac{1}{\Gamma(\alpha+1)}\Big(\frac{2N(\tau-\tau_{*})}{\rho_*}\Big)^{\alpha+1} \int_{0}^{a_N x} s^{\alpha} e^{-\frac{2N}{\rho_{*}}(\tau-\tau_{*})\cdot s} ds \,(1+ o(1))\\
&= \frac{1}{\Gamma(\alpha+1)}
\int_{0}^{b_{\tau} x} t^{\alpha} e^{-t} dt\, (1+o(1)) = \frac{\gamma(\alpha+1, b_{\tau} \,x)}{\Gamma(\alpha+1)}
\, (1+o(1)),
\end{align}
where 
$b_{\tau} = \frac{2N(\tau-\tau_{*})}{\rho_{*}} \cdot a_N $ and  
$o(1)\to 0$ uniformly in $\calK$ as $N\to \infty$. In addition, the error bound can be taken uniformly in $j$ with $m_N\leq j \leq N-1$.
Note that $b_{\tau} = O(N^{-\frac{1}{\alpha+1}}(\tau-\tau_{*}))$.
 Using the asymptotic expansion of the lower incomplete gamma function near $0$, we have
$$\gamma(\alpha+1, b_{\tau} x) = \frac{1}{\alpha+1}(b_{\tau} x)^{\alpha+1} + O(b_{\tau}^{\alpha+2}),\quad N\to \infty.$$
Thus, the integral \eqref{domint} can be estimated as
\begin{align*}
\|\zeta^j\|_{L^2(\mu_N)}^{-2}\int_{\rho_N(x)}^{\rho_{*}} 2r(\rho_{*}-r)^{\alpha}e^{h(r)}e^{-NV_{\tau}(r)} \, dr = \frac{(b_{\tau} x)^{\alpha+1}}{\Gamma(\alpha+2)}\,(1+o(1)).
\end{align*}
Taking the sum over all $j$ with $m_N \leq j \leq N-1$ as in the proof of Lemma \ref{lem:ksp}, we obtain
\begin{align*}
I_N^{\#}(x):= \sum_{j\geq m_N}^{N-1} \frac{(b_{\tau} x)^{\alpha+1}}{\Gamma(\alpha+2)} = \frac{(\alpha+2)\,x^{
\alpha+1}}{N(1-\tau_{*})^{\alpha+2}} \sum_{j\geq m_N}^{N-1} (\tau-\tau_{*})^{\alpha+1},
\end{align*}
and by the Riemann sum approximation with step length $1/N$, we have
\begin{align*}
I_N^{\#}(x)=  \frac{(\alpha+2)\,x^{
\alpha+1}}{(1-\tau_{*})^{\alpha+2}}  \int_{0}^{1-\tau_{*}} s^{\alpha+1} ds \,(1+o(1)) = x^{\alpha+1}\,(1+o(1)),\quad N\to \infty.
\end{align*}

It is a direct result from Section \ref{ssec:l} that 
all the other terms are negligible in the sum $I_N(x)$.
Indeed, by Lemma \ref{lem:sum2} and Lemma \ref{lem:sum3}, for all $x\in \calK$
\begin{align*}
\sum_{j=0}^{m_N}\int_{\calD \setminus D(0,\rho_N(x))}|p_{N,j}|^2 e^{-NQ_N} dA = o(1)
\end{align*}
as $N\to \infty$. Hence, the proof is complete.

\end{proof}

\begin{proof}[Proof of Theorem \ref{thm:3}]
From Lemma \ref{lem:IN} and \eqref{pbmamo} we have
\begin{equation*}
\log \mathbb{P}_N(\omega\geq x) = -I_N(x)+o(1) = -x^{\alpha+1} +o(1)
\end{equation*}
locally uniformly in $x$ as $N\to \infty$. Hence we prove the theorem.
\end{proof}

\section{Rescaled Ward equation at the hard wall}\label{sec:ward}

In this section, we analyze a rescaled version of limiting Ward's equation 
\begin{equation}\label{rWard}
\dbar C(z) = R(z) -\Lap \log R(z),\quad z\in \RH
\end{equation}
and prove Theorem \ref{thm:2}.

We shall explain Ward identities briefly which has been proved in \cite[Proposition 3.1]{AHM3} \cite[Theorem 4.1]{AKM}. Ward identities we consider here are  identities of correlation functions obtained from the invariance of the partition function of the particle system under the reparametrization. While Ward identities to be described below shows a global relation between correlation functions, the rescaled version of Ward equation \eqref{rWard} gives a local information of them. 
See \cite{AKM, AKMW, AKS0, AKS, AS1} for the Ward's equation approach to the local properties of Coulomb gases. We also refer to \cite[Appendix 6]{KM} for Ward's identities in the context of conformal field theory.

\subsection{Ward identities}
To set things up we borrow some notations from the literatures \cite{AHM3, AKM}.
For a test function $\psi\in C^{\infty}_0(\calD)$, we define random variables 
\begin{align*}
\I_N[\psi] = \frac{1}{2}\sum_{j\ne k}\frac{\psi(\zeta_j)-\psi(\zeta_k)}{\zeta_j-\zeta_k}, \quad \II_N[\psi] = N\sum_{j=1}^{N} \d Q_N(\zeta_j)\cdot \psi(\zeta_j),\quad \III_N[\psi]=\sum_{j=1}^{N} \d \psi(\zeta_j)
\end{align*}
where $\{\zeta_j\}_i^N$ is Coulomb particle system associated with an external potential $Q_{N}^{\calD}$. 
Consider the Hamiltonian of the system 
$$H_N(\zeta_1, \cdots, \zeta_N) = -\sum_{j\ne k}\log|\zeta_k - \zeta_j| + N\sum_{j=1}^{N}Q_N(\zeta_j).$$
The distributional form of Ward identity from \cite{AHM3, AKM} is the relation that
\begin{equation}\label{Wardid}
\E_{N,1}[\I_N[\psi]-\II_N[\psi]+\III_N[\psi]]=0,
\end{equation}
where $\E_{N,1}$ is the expectation with respect to the Boltzmann-Gibbs distribution $\Prob_{N,1}$ associated with the potential $Q_N^{\calD}$ defined in \eqref{Gbs}.
A short proof for this identity using the integration by parts from \cite[Section 4.2]{BNY} is as follows. Consider the expectation $\E_{N,1}[\d \psi(\zeta_j)]$ for a test function $\psi \in C_0^{\infty}(\calD)$. The integration by parts gives that for all $j$ 
$$\E_{N,1}[\d \psi(\zeta_j)] = \E_{N,1}[\d_{\zeta_j} H_N(\zeta_1,\cdots,\zeta_N) \cdot \psi(\zeta_j)] ,$$
which proves \eqref{Wardid}.

Now we rescale the system $\Phi^{\calD}$ at the hard wall and derive an equation satisfied by the rescaled correlation functions. With rescaling via
$$\zeta = \rho_{*} - \gamma_N z, \quad \eta = \rho_{*} - \gamma_N w,$$ we write the rescaled correlation kernel and correlation functions as
$$K_N(z,w) = \gamma_N^2 \bfK_N(\zeta, \eta)\quad \mbox{and}\quad R_{N,k}(z,w) = \gamma_N^2 \bfR_{N,k}(\zeta,\eta).$$
Write $R_{N} = R_{N,1}$ and $\bfR_{N}=\bfR_{N,1}$ for simplicity. Now we define the Berezin kernel
$$ B_N(z,w) = \frac{|K_N(z,w)|^2}{K_N(z,z)}=\frac{R_N(z)R_N(w) - R_{N,2}(z,w)}{R_N(z)},$$
and the Cauchy transform 
$$ C_N(z) = \int_{\C} \frac{B_N(z,w)}{z-w} dA(w).$$

\begin{lem}
We obtain that
$$\dbar C_N(z) = R_N(z) - \Lap \log R_N(z) + o(1),$$
where $o(1)\to 0$ locally uniformly in $\RH$ as $N\to \infty$.
\end{lem}

\begin{proof}
We apply the argument in \cite{AKM} to the hard wall case in a straightforward way. As in \cite[Theorem 7.1]{AKM}, we consider a test function $\psi$ whose support is contained in $\RH$. Then for sufficiently large $N$, the test function $\psi_N$ defined by $$\psi_N(\zeta)=\psi(\gamma_N^{-1}(\rho_{*}-\zeta))$$
has its support in $\calD$.
Now we write the expectation of random variables in the Ward identity \eqref{Wardid} in terms of the rescaled correlation functions as follows:
\begin{align*}
\E_N\,\I_N[\psi_N] 
&= \frac{1}{2}\int\!\!\!\int_{\C^2}\frac{\psi_N(\zeta)-\psi_N(\eta)}{\zeta-\eta}\, \bfR_{N,2}(\zeta,\eta) \,dA(\zeta)dA(\eta) \\ 
&=-\frac{1}{2\gamma_N}\int\!\!\!\int_{\C^2}\frac{\psi(z)-\psi(w)}{z-w} \, R_{N,2}(z,w)\, dA(z) dA(w) \\
&=-\gamma_N^{-1}\int_{\C} \psi(z) \, dA(z)\int_{\C} \frac{R_{N,2}(z,w)}{z-w}dA(w), \\
\E_N\,\II_N[\psi_N] 
&=N \int_{\C} \d Q_N(\zeta)\, \psi_N(\zeta)\, \bfR_{N}(\zeta)\,dA(\zeta)\\
&=N\int_{\C}  \d Q_N(\rho_{*} -\gamma_N z) \, \psi(z) \,R_{N}(z)\,dA(z),\\
\E_N\III_N[\psi_N] 
&= \int_{\C} \d\psi_N(\zeta)\, \bfR_{N}(\zeta)\,dA(\zeta)=-\gamma_N^{-1} \int_{\C}\d \psi(z)\,R_{N}(z)\,dA(z) \\
&= \gamma_N^{-1}\int_{\C} \d R_{N}(z)\, \psi(z) \,dA(z).
\end{align*}  
Since $\psi$ is arbitrary, we obtain from Ward identity \eqref{Wardid} that 
\begin{equation}
\int_{\C} \frac{R_{N,2}(z,w)}{z-w} dA(w) = -N\gamma_N \d Q_N(\rho_{*} - \gamma_N z)\cdot R_{N}(z) + \d R_N(z)
\end{equation}
in the sense of distributions in $\RH$.
Dividing each side by $R_N$ and taking a $\d$-derivative, we have
$$\dbar C_N(z) = R_N(z) -\Lap \log R_N(z)+ o(1)$$
since $N\gamma_N^2 \Lap Q_N(\rho_{*}-\gamma_N z) = o(1)$ locally uniformly in $\RH$ as $N\to \infty$. 
\end{proof}


\subsection{Mass-one equation and Ward's equation}
Let $K_{f}$ be a kernel of the form
$$K_{f}(z,w) = (2\re z)^{\frac{\alpha}{2}}(2\re w)^{\frac{\alpha}{2}}\int_{0}^{\infty} f(t)\,e^{-t(z+\bar{w})}\, dt,\quad z,\,w\in \RH$$
for some Borel measurable function $f$ on $\R_{+}$.

\begin{lem}
The mass-one equation 
$$\int_{\RH} |K_f(z,w)|^2 dA(w) = K_f(z,z),\quad z\in \RH$$
holds if and only if $f(s)=\Gamma(\alpha+1)^{-1} s^{\alpha+1}\cdot \1_{E}(s)$ almost everywhere for some Borel set $E$ in $\R_{+}$.
\end{lem}

\begin{proof}
Write $z=x+iy$ and $w=u+iv$ for $x,y,u,v \in \R$ and compute the integral in the mass-one equation as follows:
\begin{align*}
\int_{\RH}&|K_f(z,w)|^2 \,dA(w) 
=(2\re z)^{{\alpha}}\!\!\int_{\RH}\int_{0}^{\infty}\!\!\!\!\int_{0}^{\infty}(2\re w)^{{\alpha}} e^{-s(z+\bar{w})}e^{-t(\bar{z}+w)}f(s)f(t)\,ds \,dt \,dA(w)\\
&=\frac{(2x)^{{\alpha}}}{\pi}\!\!\int_{0}^{\infty}\!\!\!\!\int_{0}^{\infty}\!\!\!\!\int_{0}^{\infty}\!\!\!\!\int_{\R} (2u)^{{\alpha}}e^{-x(s+t)}e^{iy(t-s)}e^{-u(s+t)}e^{iv(s-t)}f(s)f(t)\,dv\,du\,ds\,dt\\
&=(2x)^{{\alpha}}\!\!\int_{0}^{\infty}\!\!\!\!\int_{0}^{\infty} 2 (2u)^{{\alpha}}e^{-2x \cdot s} e^{-2s \cdot u}\, f(s)^2 du\,ds = (2x)^{\alpha}\int_{0}^{\infty} e^{-2x\cdot s}\, f(s)^2 \frac{\Gamma(\alpha+1)}{s^{\alpha+1}}\,ds.
\end{align*}
Thus, the mass-one equation holds if and only if for all $x>0$
$$\int_{0}^{\infty} e^{-2x\cdot s} f(s) ^2 \frac{\Gamma(\alpha+1)}{s^{\alpha+1}}\, ds = \int_{0}^{\infty} e^{-2x\cdot s} f(s) \,ds,$$
which implies that $f(s) =\Gamma(\alpha+1)^{-1} s^{\alpha+1}\1_{E}(s)$ almost everywhere for some Borel measurable set $E$ in $\R_{+}$.
\end{proof}

\begin{lem}
The rescaled Ward's equation 
\begin{equation*}
\dbar C(z) = R(z) -\Lap \log R(z),\quad z\in \RH
\end{equation*}
holds for $K_f$ if and only if there exists a connected set $E$ in $\R_{+}$ such that 
$f(s) = \Gamma(\alpha+1)^{-1} s^{\alpha+1} \cdot \1_{E}(s)$ almost everywhere in $\R_{+}$.
\end{lem}
\begin{proof}
First 
we compute the Cauchy transform as follows: for $z=x+iy$
\begin{align*}
R(z)\, C(z)
&=\int_{\RH} |K_f(z,w)|^2 \frac{1}{z-w}\, dA(w) \\
&=\frac{(2x)^{\alpha}}{\pi}\!\int_{0}^{\infty}\!\!\!\!\int_{0}^{\infty}\!f(s)\,f(t)\!\int_{-\infty}^{\infty}\int_{-x}^{\infty} (2(u+x))^{\alpha} e^{-(s+t)(u+2x)}\,\frac{e^{i(s-t)v}}{-(u+iv)}\,du\,dv \, ds\,dt.
\end{align*}
Using the fact that
$$\int_{-\infty}^{\infty}\frac{e^{i(s-t)v}}{u+iv}\frac{dv}{\pi}=
\begin{cases}
-2 e^{u(t-s)}\,\1_{t>s},\quad & \mathrm{if}\quad u<0,\\
2e^{u(t-s)}\,\1_{t<s}, \quad & \mathrm{if} \quad u>0,
\end{cases}$$
we obtain that
$R(z)C(z)=-L_1(z) + L_2(z),$ where
\begin{align*}
L_1(z) &= (2x)^{\alpha}\int_{0}^{\infty}\!\!\!\int_{0}^{\infty} f(s)f(t) \!\int_{0}^{\infty} 2(2(u+x))^{\alpha}e^{-2x(s+t)}e^{-2s\cdot u}du \, ds\,dt,\end{align*}
and
\begin{align*}
L_2(z) &= (2x)^{\alpha}\int_{0}^{\infty}\!\!\!\int_{0}^{t}f(s)f(t)\!\int_{-x}^{\infty}2(2(u+x))^{\alpha}e^{-2x(s+t)} e^{-2s\cdot u} \,du\, ds\,dt.
\end{align*} 
A simple calculation gives  
\begin{align*}
L_1(z) &= (2x)^{\alpha} \int_{0}^{\infty}\!\!\!\int_{0}^{\infty} f(s)f(t)\,e^{-2x\cdot t} s^{-(\alpha+1)} \Gamma(\alpha+1, 2xs) \,ds\,dt \\
&= R(z)\int_{0}^{\infty} f(s)\,s^{-(\alpha+1)} \Gamma(\alpha+1, 2xs)\,ds,
\end{align*}
and
\begin{align*}
L_2(z) &= (2x)^{\alpha}\int_{0}^{\infty}\!\!\!\int_{0}^{t} f(s)f(t)\, e^{-2x\cdot t} s^{-(\alpha+1)} \Gamma(\alpha+1) \,ds\,dt.
\end{align*}
Here, $\Gamma(\alpha+1,x)=\int_{x}^{\infty} t^{\alpha} e^{-t} dt$ is the upper incomplete gamma function.
Since $L_1$, $L_2$, and $R$ are functions of $x=\re z$ only, we see that
\begin{align*}
-\dbar\left(\frac{L_1}{R}\right)=\int_{0}^{\infty} f(s) (2x)^{\alpha}\,e^{-2x\cdot s} ds =R(z).
\end{align*}
Thus, the Ward equation is equivalent to the relation
\begin{align*}
\dbar\left(\frac{L_2}{R}\right) = \dbar\left(\frac{L_1}{R} + C\right)= -R + \dbar{C} = - \dbar \left(\frac{\d R}{R}\right),
\end{align*}
which implies that
$$L_2 + \d R = c R$$
for some constant $c$. This relation can be written as
\begin{align*}
\int_{0}^{\infty}\!\!\!\int_{0}^{t} &f(s) f(t) \, e^{-2x\cdot t}s^{-(\alpha+1)}\Gamma(\alpha+1)\,ds\,dt = \int_{0}^{\infty}  (t+c)\,f(t)\, e^{-2x\cdot t} dt 
\end{align*}
so that we obtain
$$\int_{0}^{t}f(s) s^{-(\alpha+1)}\Gamma(\alpha+1)\,ds = t+c$$
when $f(t)\ne 0$. This gives $f(t)= \Gamma(\alpha+1)^{-1} t^{\alpha+1}\cdot \1_{E}(t)$ a.e. for some subset $E$ of $\R_{+}$. Hence $E$ is connected since for each $t\in E$,
$|E \cap (0,t) | = t+c$ where $|E|$ denotes Lebesgue measure of a set $E\subset \R$.
\end{proof}

\subsection*{Acknowledgements}

The author thanks Nam-Gyu Kang for valuable advice and discussions. This work was partially supported by the KIAS Individual Grant (MG063103) at Korea Institute for Advanced Study and by the National Research Foundation of Korea (NRF) Grant funded by the Korea government (2019R1F1A1058006).


\end{document}